\documentclass[letterpaper, 10 pt, conference]{ieeeconf}

\usepackage{xcolor}
\usepackage[utf8]{inputenc}
\usepackage{amsmath}
\usepackage{amssymb}
\usepackage{graphicx,balance}
\usepackage{pifont}
\usepackage{bm}
\usepackage{dsfont}
\usepackage{subcaption}

\usepackage{siunitx}

\usepackage{amsfonts}
\usepackage{url}
\usepackage[pdftex,plainpages = false, colorlinks=true, linkcolor=black, citecolor = black, urlcolor = blue,pagebackref=false,hypertexnames=false, plainpages=false, pdfpagelabels]{hyperref}
\usepackage{balance}
\usepackage[capitalize]{cleveref}
\usepackage[sort,compress]{cite}

\crefname{section}{Section}{Sections}
\crefname{theorem}{Theorem}{Theorems}
\crefname{lemma}{Lemma}{Lemmas}
\crefname{table}{Table}{Tables}
\crefformat{equation}{(#2#1#3)}
\crefname{algocf}{Algorithm}{Algorithms}
\Crefname{algocf}{Algorithm}{Algorithms}
\crefname{ALC@unique}{Line}{Lines}

\newtheorem{thm}{Theorem}
\newtheorem{lem}{Lemma}

\newtheorem{asm}{Assumption}
\crefname{asm}{Assumption}{Assumptions}
\Crefname{asm}{Assumption}{Assumptions}

\newcommand{\x}{x}
\newcommand{\s}{s}

\renewcommand{\u}{u}
\renewcommand{\d}{d}
\newcommand{\mDx}{\mathcal{D}}
\newcommand{\z}{z}
\newcommand{\y}{y}
\newcommand{\e}{e}
\newcommand{\iss}{\overline{\Omega}^*}
\newcommand{\uprf}{\u}
\newcommand{\cvx}{\bar{\mathcal{S}}_t}
\newcommand{\p}{p}
\renewcommand{\v}{v}

\newcommand{\nn}{L_{\theta}}

\usepackage{accents}

\usepackage{paralist}

\crefformat{chapter}{\S#2#1#3}
\crefmultiformat{chapter}{\S\S#2#1#3}{and~#2#1#3}{, #2#1#3}{, and~#2#1#3}
\crefformat{section}{\S#2#1#3}
\crefmultiformat{section}{\S\S#2#1#3}{and~#2#1#3}{, #2#1#3}{, and~#2#1#3}

\usepackage{algorithm,algorithmic}

\usepackage{tikz}
\usetikzlibrary{shapes,arrows,arrows.meta,positioning,calc}

\newcommand{\gdn}{GUARDIAN}
\definecolor{review}{RGB}{0, 0, 0}
\newcommand{\review}[1]{{\color{review}#1}}
\definecolor{revieww}{RGB}{0, 0, 0}
\newcommand{\revieww}[1]{{\color{revieww}#1}}

\title{GUARDIAN: Safety Filtering for Systems with Perception Models Under Adversarial Attack}
\author{Nicholas Rober, Alex Rose, and Jonathan P.\ How
}

\begin{document}

\newif\ifarxiv
\arxivfalse

\maketitle
\thispagestyle{empty}
\setlength{\skip\footins}{8pt}
\begingroup
\renewcommand{\thefootnote}{}
\footnotetext{Submitted February $4^\mathrm{th}$, 2026.
This material is based upon work supported by the Naval Information Warfare Center (NIWC) Atlantic under Contract No. N6523623C8011.
The views, opinions, and/or findings expressed are those of the author(s) and should not be interpreted as representing the official views or policies of the Department of Defense or the U.S. Government.}
\footnotetext{Code: \url{https://github.com/mit-acl/guardian}}
\endgroup

\newif\ifshortabs
\shortabstrue
\begin{abstract}
  \ifshortabs
  Safety filtering is an effective method for enforcing constraints in safety-critical systems, but existing methods typically assume perfect state information.
  This limitation is especially problematic for systems that rely on neural network (NN)-based state estimators, which can be highly sensitive to noise and adversarial input perturbations.
  We address these problems by introducing \gdn: Guaranteed Uncertainty-Aware Reachability Defense against Adversarial INterference, a safety filtering framework that provides formal safety guarantees for systems with NN-based state estimators.
  At runtime, \gdn\ uses neural network verification tools to provide guaranteed bounds on the system's state estimate given possible perturbations to its observation.
  It then uses a modified Hamilton-Jacobi reachability formulation to construct a safety filter that adjusts the nominal control input based on the verified state bounds and safety constraints.
  The result is an uncertainty-aware filter that ensures safety despite the \review{use of} an NN estimator with noisy, possibly adversarial, observations.
  Theoretical \review{and} numerical \review{results show how} \gdn\ effectively defends systems against adversarial attacks that would otherwise \review{cause} violation of safety constraints.
  
  \else
  Safety filtering is an effective method for enforcing constraints in safety-critical systems, but existing methods typically assume perfect state information.
  This limitation is especially problematic for systems that rely on neural network (NN)-based state estimators, which can be highly sensitive to noise and adversarial input perturbations.
  We address these problems by introducing \gdn: Guaranteed Uncertainty-Aware Reachability Defense against Adversarial INterference, a safety filtering framework that provides formal safety guarantees for systems with NN-based state estimators.
  At runtime, \gdn\ uses neural network verification tools to provide guaranteed bounds on the system's state estimate given possible perturbations to its observation.
  It then uses a modified Hamilton-Jacobi reachability formulation to construct a safety filter that adjusts the nominal control input based on the verified state bounds and safety constraints.
  The result is an uncertainty-aware filter that ensures safety despite the system's reliance on an NN estimator with noisy, possibly adversarial, input observations.
  Theoretical analysis and numerical experiments demonstrate that \gdn\ effectively defends systems against adversarial attacks that would otherwise lead to a violation of safety constraints.
  \fi
  \vspace{-4pt}
\end{abstract}


\section{Introduction}
\label{sec:introduction}
Neural network (NN)-based machine learning continues to demonstrate its utility as a tool for building autonomous systems with impressive capabilities~\cite{singh2022reinforcement}. 
However, NNs are sensitive to input perturbations~\cite{szegedy_intriguing_2014}, posing a significant safety risk associated with their application~\cite{bak_zero-one_2024}.
This threat is exacerbated for systems with NN-based perception modules, i.e., observation-based neural feedback loops (ONFLs), because adversarial agents can manipulate system-level behavior by attacking the perception model's observations~\review{\cite{arnstrom2024data,robey_jailbreaking_2025}}. 
Thus, it is important to develop protective methods for systems with learned perception modules subject to adversarial attacks before they can be applied in safety-critical scenarios.

\begin{figure}
    \centering
    \includegraphics[width=0.66\linewidth]{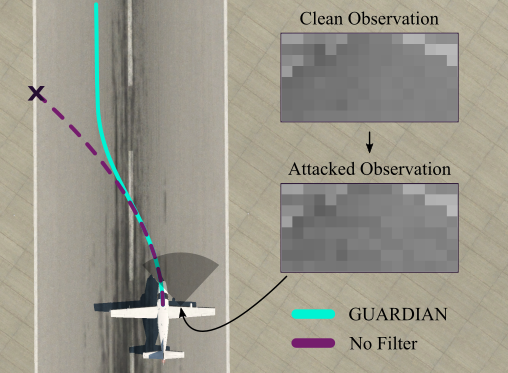}
    \caption{\gdn\ protects against adversarial attacks.}
    \label{fig:guardian_intro}
    \vspace{-20pt}
\end{figure}

In response to the threat posed by adversarial attacks, NN verification (NNV) has become an active area of research.
A wide range of open-loop verification techniques~\cite{zhang_efficient_2018,xu_automatic_2020,jaxverify} have been developed to calculate guaranteed bounds on an NN output given a nominal input and a set of possible perturbations.
These ideas have been extended to design-time verification of closed-loop systems, primarily via set-propagation-based reachability analysis techniques, for systems with NN controllers~\cite{rober_constraint-aware_2024,wang_verifying_2025} and, less frequently, ONFLs~\cite{katz_verification_2022,lin_robust_2025}.
While \cite{katz_verification_2022,lin_robust_2025} shows how NNV can be used for ONFLs prior to deployment, it is difficult to incorporate a priori unknown perturbations from adversarial attacks that are infused at runtime.
Recently,~\cite{lin_robust_2025} accomplished this by employing Hamilton-Jacobi (HJ) reachability analysis, but their approach is dependent on an abstraction of the closed-loop system that is difficult to obtain for complicated control pipelines. 
Thus, to protect ONFLs from adversarial attacks, we instead consider runtime safety assurances.

Safety filtering~\cite{hsu_safety_2024,wabersich_data-driven_2023}, i.e., modifying nominal control inputs to ensure satisfaction of safety constraints, is a common approach to ensure safety at runtime.
Recent work has shown that NNV tools can be used as a runtime safety monitor~\cite{mahesh_safe_2025}, but they have not yet been demonstrated as a means of safety filtering.
Instead, most safety filters are constructed using either control barrier functions (CBFs)~\cite{ames2016control} or Hamilton-Jacobi reachability analysis~\cite{bansal_hamilton-jacobi_2017}.
While most work on safety filtering assumes perfect state information, several recent works have focused on incorporating observation models and accounting for imperfect state information~\review{\cite{tan2024safety,tan2025secure,dean_guaranteeing_2021,nanayakkara_safety_2025,das_safe_2025,lin2026stochastic}}.
\review{Namely, \cite{tan2024safety,tan2025secure} develop CBF-based defenses for sensor-spoofing attacks on linear systems.} 
\review{For control-affine nonlinear systems,} \cite{dean_guaranteeing_2021} proposed a measurement-robust (MR)-CBF formulation that accounts for state-dependent measurement error via a more restrictive CBF condition.
Furthermore, \cite{nanayakkara_safety_2025} developed an alternative robust (R)-CBF method that similarly accounts for measurement error and is guaranteed to exist if there \review{is a standard} CBF, but does not capture state-dependent estimation error.
Additionally, \review{in the presence of large estimation errors}, R-CBFs are proven to be invariant for an inflated region around the original safe set, leading to conservative behavior or possible safety violations \review{under strong adversarial attacks}.
Finally, \cite{das_safe_2025} extends \cite{nanayakkara_safety_2025} by adaptively tuning an R-CBF to be less conservative and capture state-dependent error, but again is invariant for an inflated version of the safe set.
Moreover, while \cite{dean_guaranteeing_2021,nanayakkara_safety_2025,das_safe_2025} all incorporate measurement error, determining that error in the context of adversarial attacks has not yet been explored.

In contrast to \cite{dean_guaranteeing_2021,nanayakkara_safety_2025,das_safe_2025}, this paper proposes an HJ-reachability-based safety filter, \gdn: Guaranteed Uncertainty-Aware Reachability Defense against Adversarial INterference.
\gdn\ uses state-of-the-art NNV tools to compute state uncertainty given a nominal observation and a bound on the strength of possible adversarial attacks.
This uncertainty set is then incorporated into a modified HJ-reachability formulation that filters potentially unsafe nominal control inputs, ensuring provably safe closed-loop behavior under adversarial influence.
We present several examples illustrating key features of \gdn, including its ability to protect image-based ONFLs, its handling of state-dependent vulnerability to adversarial attacks, and its performance relative to MR-CBFs~\cite{dean_guaranteeing_2021}, R-CBFs~\cite{nanayakkara_safety_2025}, and R-CBF-QPs~\cite{das_safe_2025}. 
%
%
The main contributions of this work include:
\begin{itemize}
  \item A modified HJ reachability formulation that incorporates state-estimation uncertainty, enabling safety guarantees for \review{control-affine} single-input systems with estimation modules subject to \review{perturbed} observations.
  \item \gdn: a unified safety-filtering framework that combines our modified HJ reachability with modern NNV techniques to protect systems with NN-based perception modules from adversarial attacks.
  \item Numerical examples demonstrating \gdn's effectiveness in safeguarding ONFLs against adversarial manipulation and comparing its performance to CBF-based approaches~\cite{dean_guaranteeing_2021,nanayakkara_safety_2025,das_safe_2025}.
\end{itemize}

\section{Preliminaries}
\label{sec:prelims}
\subsection{System Dynamics}
\label{sec:dynamics}
Consider the nonlinear discrete-time system with state vector $\x_{\color{review}t} \in \mathcal{X} \subseteq \mathbb{R}^{n_x}$ and control-affine dynamics
\begin{equation}
\label{eqn:dynamics}
    \x_{t+1} = f_d(\x_t, \d_t) +  g(\x_t)\u_t \triangleq f(\x_t, \u_t, \d_t)
\end{equation}
where $\u_t \in [\underline{u}, \overline{u}] \triangleq \mathcal{U} \subset \mathbb{R}$ is the scalar control input, and $\d_t \in \review{\mDx} \subset \mathbb{R}^{n_x}$ is the disturbance term capturing the predictive uncertainty of the system due to unknown external forces, actuator degradation, and modeling error.
While our approach may work in practice for multidimensional $\u$, we consider the scalar case here due to improved interpretability of the theoretical argument presented in \cref{sec:theory}.
The system has a sensor providing measurements $\y_t \in \mathcal{Y} \subseteq \mathbb{R}^{n_y}$ from the environment according to the function $\y_t = h(\x_t, \d^y_t)$ where $\d^y_t \in \mathcal{D}^y$ represents possible sensor noise.
The sensor measurements are passed to an \review{NN estimator $\nn: \mathcal{Y} \rightarrow \mathcal{X}$ that provides an estimate} of the system's state $\hat{\x}_t$, i.e., $\hat{\x}_t = \nn(\y_t)$, \review{and is assumed to have an error bounded by $\e_{\hat{\x}} \in \mathbb{R}_{\geq0}^{n_x}$, i.e., $|\hat{\x}_{t,i} - \x_{t,i}| \leq \e_{\hat{\x},i}\ \forall i\in\{1\ldots n_x\}$, which can be determined and verified in a post-training calibration stage.}
The system also has a controller $K_d: \mathcal{X} \rightarrow \mathcal{U}$, which generates an input $\u_t$ based on the state estimate from $\nn$.
Finally, the system is subject to safety constraints encoded by designing a function $c: \mathcal{X} \rightarrow \mathbb{R}$ (e.g., a signed distance function) such that $\mathcal{C} \triangleq \{\x\ |\ \review{c(\x) \geq 0}\}$ is the set of \review{safe} states.

\subsection{Adversarial Attacks}
Adversarial attacks are constructed perturbations to an NN's input that cause it to behave in an undesirable way.
\review{The results presented in \cref{sec:numerical_results} consider targeted, whitebox attacks, but the approach proposed in \cref{sec:approach} works for general $\epsilon$-bounded perturbations.}
\review{The goal of the attacker is to find an adversarial perturbation $\Delta\y \in \mathbb{R}^{n_y}$} applied to the input of an NN $\pi: \mathcal{Y} \rightarrow \mathcal{X}$ \review{with known structure and parameters, where ${\tilde{\y} \triangleq \y + \Delta\y}$.
The perturbation is constrained to satisfy $\|\Delta\y\|_\infty \leq \epsilon$ for some maximum perturbation $\epsilon \in \mathbb{R}_{>0}$, and} is chosen to approximately minimize the distance between a target $\x^*\in\mathbb{R}^{n_x}$ and the adversarially perturbed NN output ${\tilde{\hat{\x}} = \pi(\tilde{\y})}$.
\review{More detail on generating these attacks is given in \cref{sec:numerical_results}.}
As will be shown, \review{such adversarial attacks can trick an ONFL into violating safety constraints if it is not equipped with a safety filter designed to protect against them}.

\subsection{Neural Network Verification}
Neural network verification (NNV) refers to a set of techniques that formally analyze the computation patterns of a given NN to make statements about its behavior as follows:
\begin{thm}[NN Robustness Verification~\cite{zhang_efficient_2018}] \label{thm::nnv}
  Given a neural network $\pi: \mathbb{R}^{n_z} \rightarrow \mathbb{R}^{n_o}$ and a hyper-rectangular set of possible inputs $\mathcal{Z}$, there exist two explicit functions
  \begin{equation}
    \underline{\pi}(\z) = \Psi \z + \alpha, \quad \overline{\pi}(\z) = \Xi \z + \beta
  \end{equation}
such that the inequality $\underline{\pi}(\z) \leq \pi(\z) \leq \overline{\pi}(\z)$ holds element-wise for all $z \in \mathcal{Z}$, with $\Psi, \Xi \in \mathbb{R}^{n_o \times n_z}$ and $\alpha, \beta \in \mathbb{R}^{n_o}$.
\end{thm}
\cref{thm::nnv} is proven in \cite{zhang_efficient_2018} and implemented so that the bounds $\underline{\pi}(\z)$ and $\overline{\pi}(\z)$ can be found using tools such as \texttt{CROWN}~\cite{zhang_efficient_2018}, \texttt{auto\_LiRPA}~\cite{xu_automatic_2020}, and \texttt{jax\_verify} \cite{jaxverify}.

\subsection{Hamilton-Jacobi Reachability Analysis}
\label{sec:prelims:hj_reach}
Hamilton-Jacobi (HJ) reachability analysis is an optimal control approach that poses reachability as a two-player game with worst-case disturbances pushing the system out of the safe set and control inputs acting in opposition, attempting to keep the system safe \review{\cite{bansal_hamilton-jacobi_2017}}.
In discrete-time HJ reachability analysis, the solution to the Isaacs equation~\cite{isaacs1954differential} can be obtained from dynamic programming via
\begin{align}
\label{eqn:programming}
    V_k(\x) = \min \left(c(\x),\ \max_{\u\in\mathcal{U}}\min_{\d\in\review{\mDx}}V_{k+1}(f(\x, \u, \d))\right)\\ 
    k\in\{0, 1, \ldots, T-1\},\ V_T(\x) = c(\x). \nonumber
\end{align}
The value function $V(\x) = \lim_{k\rightarrow\infty}V_k(\x)$ encodes the set of states for which there exists some action that the system can take to guarantee safety from disturbances.
Conversely, if $V(\x) < 0$, there is some sequence of disturbances that can cause a safety violation, i.e., $c(\x_T) \leq 0$, for any sequence of control inputs.
Using $V(\x)$, we can construct the maximal safe set $\Omega^* \triangleq \{\x\in\mathcal{X} \ |\ V(\x) \geq 0\}$ with a boundary $\partial \Omega^* \triangleq \{\x\in\mathcal{X} \ |\ V(\x) = 0\}$
\review{and the one-step control-disturbance-invariant safe set $\iss \triangleq \{\x\in\mathcal{X} \ |\ \min_{\u\in\mathcal{U},\d\in\mDx} V(f(\x, \u, \d)) \geq 0\}$.}
\review{The optimal} safe control policy \review{then becomes} $\u^*_{HJ}(\x)$:
\begin{equation}
\label{eqn:nominal_hj}
    \u^*_{HJ}(\x) \in \arg\max_{\u\in\mathcal{U}}\min_{\d\in \review{\mDx}}V(f(\x, \u, \d)).
\end{equation}
An HJ-based safety filter can then be constructed \review{as}
\begin{equation}\label{eqn:hj}
    K_{F,\mathrm{nom}}(\x, \u) \triangleq
    \begin{cases}
        \u, \quad & \min_{\d\in\review{\mDx}}V(f(\x, \u, \d)) \geq 0 \\
        \u^*_{HJ}, \quad & \mathrm{otherwise},
    \end{cases}
\end{equation}
which ensures that if the unfiltered input $\u$ causes the system to leave $\Omega^*$, then the safety filter will be activated.
\section{\gdn}
\label{sec:approach}
\newcommand{\dummy}{\s}
As introduced in \cref{sec:introduction}, \gdn\ is a safety filtering approach that incorporates measurement uncertainty determined via NNV into a modified HJ reachability formulation to design a safe control policy.
The first key idea in \gdn\ is to use NNV to determine a set containing all possible states $\bar{\mathcal{X}}_t$ given an adversarially perturbed $\tilde{\y}_t$.
Given that the perturbation $\|\Delta\y\review{_t} \|_\infty \leq \epsilon$, we know the true observation is contained by the set $\bar{\mathcal{Y}}_t \triangleq \{\y_t \ |\ \|\y_t - \tilde{\y}_t\|_\infty \leq \epsilon\}$. 
Applying \cref{thm::nnv}, we can then use NNV tools \review{at runtime} to generate upper and lower bounds on the output of $\nn$, i.e., $\overline{\tilde{\x}}_t$ and $\underline{\tilde{\x}}_t$, respectively, such that $\tilde{\hat{\x}}_t \in [\underline{\tilde{\x}}_t,\ \overline{\tilde{\x}}_t]$.
\review{Given the state estimate bound $\e_{\hat{\x}}$, we then construct}
\begin{equation}
    \label{eqn:gdn_nnv}
    \bar{\mathcal{X}}_t \triangleq [\underline{\tilde{\x}}_t-\e_{\hat{\x}},\ \overline{\tilde{\x}}_t+\e_{\hat{\x}}],
\end{equation}
where it can then be shown that $\x_t \in \bar{\mathcal{X}}_t$.
Defining ${\Phi(\bar{\mathcal{X}}_t,\u) \triangleq \min_{\dummy \in \bar{\mathcal{X}}_t,\d\in\review{\mDx}}V(f(\dummy, \review{\u}, \d))}$, reformulate \cref{eqn:hj}:
\begin{equation}
\label{eqn:gdn}
    K_F(\bar{\mathcal{X}}_t, \u_t) \triangleq
    \begin{cases}
        \u_t, \quad & \Phi(\bar{\mathcal{X}}_t,\u_t) \geq 0 \\
        \u^*_t, \quad & \mathrm{otherwise},
    \end{cases}
\end{equation}
where the optimal safe control policy \review{is now}
\begin{equation}
    \label{eqn:gdn_ctrl}
    \u^*_t(\bar{\mathcal{X}}_t) \in \arg\max_{\u\in\mathcal{U}}\ \Phi(\bar{\mathcal{X}}_t,\u).
\end{equation}

\newif\ifexpand
\expandfalse
\definecolor{expanded1}{RGB}{200, 0, 0}
\definecolor{expanded2}{RGB}{20, 20, 150}


\subsection{Theoretical Results}

\label{sec:theory}
To validate the safety of \gdn, we then introduce \cref{asm:safe_cone_scalar,asm:uncertainty_bound_scalar}, along with \cref{lem:monotone_improvement_scalar} and \cref{thm:gdn_scalar}.
\review{The assumptions} provide sufficient conditions to ensure there is a uniform safe control direction for any $\dummy \in \bar{\mathcal{X}}_t$ and provide that $\bar{\mathcal{X}}_t$ is a reasonable size relative to $\Omega^*$.
\begin{asm}
    \label{asm:safe_cone_scalar}
    There exists a neighborhood $\mathcal N(\partial\Omega^*) \supset \partial\Omega^*$ such that $V\in C^1(\mathcal{N}(\partial\Omega^*))$ and for any convex $\mathcal{S} \subset \mathcal N(\partial\Omega^*)$, there is a cone $\mathcal{K}(\mathcal{S)} \review{\subseteq \mathbb{R}}$ \review{of safe controls,} where
    \vspace{-0.3\baselineskip}
    \begin{equation}
        \label{eqn:gradient_condition}
        \nabla V(\review{f(\dummy, \u, \d)})^\top g(\dummy) \uprf \geq 0\ \forall \dummy \in \mathcal{S}, \review{\forall \d \in \review{\mDx},} \forall \uprf \in \mathcal{U}\review{\cap\mathcal{K}(\mathcal{S})}.
    \end{equation}
    \vspace{-1.3\baselineskip}
\end{asm}
\begin{asm}
    \label{asm:uncertainty_bound_scalar}
    The attack strength $\epsilon$ and error term $\e_{\hat{\x}}$ are \review{small enough that if $\Phi(\bar{\mathcal{X}}_t,\u_t) < 0$, then} $\bar{\mathcal{X}}_t \subset \mathcal{N}(\partial \Omega^*)\review{\cup\iss}$ 
    \review{and there is a convex $\mathcal{S}\subset\mathcal{N}(\partial \Omega^*)$ such that $(\bar{\mathcal{X}}_t \setminus \iss) \subseteq \mathcal{S}.$}
\end{asm}

Note that \cref{asm:safe_cone_scalar} can be checked offline while \cref{asm:uncertainty_bound_scalar} must be checked at runtime due to the complicated relationship between $\nn$ and possible input perturbations.
As shown in \cref{sec:numerical_results}, these assumptions allow \review{for} the application of \gdn\ in a variety of situations.
\begin{lem}
\label{lem:monotone_improvement_scalar}
    Let \review{\cref{asm:safe_cone_scalar}} hold for $\mathcal{S} \subset \mathcal{N}(\partial\Omega^*)$ and $\uprf_1, \uprf_2 \in \mathcal{U}$ satisfy $\uprf_2 - \uprf_1 \in \mathcal{K}(\mathcal{S})$.
    Then $V(f(\dummy, \uprf_1, \d)) \leq V(f(\dummy, \uprf_2, \d))\ \forall \dummy \in \mathcal{S},\ \forall \d \in \review{\mDx}$.
\end{lem}
The proof for \cref{lem:monotone_improvement_scalar} follows directly from integrating the monotonicity condition \cref{eqn:gradient_condition} from $\uprf_1$ to $\uprf_2$.
\ifexpand

{\color{expanded2}
\begin{proof}
    Fix an arbitrary $\dummy \in \mathcal{S}$ and $\d \in \review{\mDx}$ and define the scalar function $\psi(\alpha) \triangleq V(f(\dummy, \u_1 + \alpha (\u_2 - \u_1), \d))$, $\forall\alpha \in [0,1]$ specifying the value of the next state interpolating between $\u_1$ and $\u_2$.
    Since $V$ is $C^1$ on $\mathcal{N}(\partial\Omega^*)$ and $f$ is affine in $\u$, $\psi$ is differentiable and
    \begin{equation*}
        \psi'(\alpha) = \nabla V(f(\dummy, \u_1 + \alpha(\u_2 - \u_1), \d))^\top g(\dummy)(\u_2 - \u_1).
    \end{equation*}
    By assumption, $\u_2 - \u_1 \in \mathcal{K}(\mathcal{S})$, so given \cref{asm:safe_cone_scalar,asm:uncertainty_bound_scalar},
    \begin{equation*}
    \label{eqn:notation_issue}
        \begin{split}
            \nabla V(f(\mathbf{\xi}, \u_1 + \alpha(\u_2 - \u_1), \d))^\top g(\mathbf{\xi})(\u_2 - \u_1) \geq 0&\\ 
            \forall \mathbf{\xi} \in \mathcal{S}\ \forall \mathbf{\d} \in \review{\mDx},&
        \end{split}
    \end{equation*}
    implying $\psi(\alpha)' \geq 0\ \forall \alpha\in[0,1]$.
    As a result, integrating along $\alpha$ results in
    \begin{equation*}
        \psi(1) - \psi(0) = \int_0^1\psi'(\alpha)d\alpha \geq 0,
    \end{equation*}
    i.e., $V(f(\dummy, \u_2, \d)) \geq V(f(\dummy, \u_1, \d))$.
    Since $\x$ and $\d$ were chosen arbitrarily, $V(f(\dummy, \u_2, \d)) \geq V(f(\dummy, \u_1, \d))\ \forall \dummy \in \mathcal{S},\ \forall \d \in \review{\mDx}$, thus completing the proof.
\end{proof}}\fi
\begin{thm}
\label{thm:gdn_scalar}
    Consider a system with dynamics \cref{eqn:dynamics}, safe set $\Omega^*$, and true initial state $\x_0\in\Omega^*$.
    If \cref{asm:safe_cone_scalar,asm:uncertainty_bound_scalar} hold,
    then for each $t \in \mathbb{Z}_{\geq 0}$, constructing $\bar{\mathcal{X}}_t$ as specified in \cref{eqn:gdn_nnv} and applying the safety filter \cref{eqn:gdn} results in $\x_t\in\Omega^*$, thereby rendering $\Omega^*$ an invariant set.
\end{thm}

\begin{proof}
    We proceed by induction, showing that $\x_t \in \Omega^* \implies \x_{t+1} \in \Omega^*$, where the base case $\x_0 \in \Omega^*$ holds by assumption.
    Then, for the inductive step, suppose $\x_t \in \Omega^*$ and recognize that the filter is either $(i)$ inactive (first case in \cref{eqn:gdn}), or $(ii)$ active (second case \cref{eqn:gdn}).
    For the inactive case, the nominal control value $\u_t$ results in $\Phi(\bar{\mathcal{X}}_t,\u_t) \geq 0$, so $\forall \dummy \in \bar{\mathcal{X}}_t,\forall\d\in\review{\mDx}, V(f(\dummy,\u_t,\d)) \geq 0$.
    Since $V(s)\geq 0 \iff s\in \Omega^*$, $\x_{t+1}\in\Omega^*$ follows for the inactive case.

    For the active case, $\Phi(\bar{\mathcal{X}}_t,\u_t) < 0$, so by \cref{asm:uncertainty_bound_scalar} we have $\bar{\mathcal{X}}_t \subset \mathcal{N}(\partial \Omega^*)\cup\iss$.
If $\dummy\in\iss$, then $f(\dummy,\u,\d) \in \Omega^*\, \forall\u \in \mathcal{U}, \forall\d\in\mDx$, so we may instead consider a convex set $\cvx \supseteq \bar{\mathcal{X}}_t \setminus \iss$ satisfying $\cvx\subset\mathcal{N}(\partial \Omega^*)$ and whose existence is ensured by \cref{asm:uncertainty_bound_scalar}.
Furthermore, \revieww{$\u^*(\cvx)=\u^*(\bar{\mathcal{X}}_t)\triangleq\uprf^*$ since the minimizing state of $\Phi$ cannot lie in $\iss$}.
If calculating $\uprf^*$ via \cref{eqn:gdn_ctrl} results in $\Phi(\cvx, \uprf^*) \geq 0$, then, similar to the inactive case, we get $\x_{t+1}\in\Omega^*$ directly.
Alternatively, \revieww{suppose $\Phi(\cvx, \uprf^*) < 0$.
Since $\x_t \in \Omega^*$,} by construction of $\Omega^*$, $\exists \u^\dagger(\x_t) \in \mathcal{U}$ such that $\min_{\d \in \mDx} V(f(\x_t,\u^\dagger(\x_t),\d)) \geq 0$~\cite{isaacs1954differential}.
\revieww{We now show that $\uprf^* - \uprf^\dagger(\x_t) \in \mathcal{K}(\cvx)$ regardless of which form $\mathcal{K}(\cvx)$ takes.}

From \cref{asm:safe_cone_scalar}, $\mathcal{K}(\cvx) \in \{\mathbb{R}_{\leq0}, \mathbb{R}_{\geq0}, \mathbb{R}\}$ for scalar $\u$.
\revieww{Consider $\mathcal{K}(\cvx) = \mathbb{R}_{\geq 0}$: the monotonicity condition \cref{eqn:gradient_condition} implies $V(f(\dummy,\u,\d))$ is non-decreasing in $\u$ over $\cvx$, so the maximizer in \cref{eqn:gdn_ctrl} lies at the upper bound of $\mathcal{U}$, i.e., $\uprf^* = \overline{u}$.}
\revieww{Since $\uprf^\dagger(\x_t) \in \mathcal{U}$, $\uprf^* \geq \uprf^\dagger(\x_t)$, so $\uprf^* - \uprf^\dagger(\x_t) \in \mathbb{R}_{\geq 0} = \mathcal{K}(\cvx)$}.
\ifexpand
{\color{expanded2}
Suppose instead $\u^* \neq \overline{u}$.
Then $\overline{u} - \u^* \in \mathcal{K}(\cvx)$, so \cref{lem:monotone_improvement_scalar} gives
\begin{equation}
    \label{eqn:show_u_max}
    V(f(\dummy, \u^*, \d)) \leq V(f(\dummy, \overline{u}, \d))\ \forall \dummy \in \cvx,\ \forall \d \in \mDx.
\end{equation}
In particular, $\Phi(\cvx, \u^*) \leq \Phi(\cvx, \overline{u})$, but by definition $\Phi(\cvx, \u) \leq \Phi(\cvx, \u^*)\ \forall \u \in \mathcal{U}$, so $\Phi(\cvx, \u^*) = \Phi(\cvx, \overline{u})$ and $\u^* = \overline{u}$, a contradiction.
}
\fi
A similar argument can be made for $\mathcal{K}(\cvx) = \mathbb{R}_{\leq 0}$.
\revieww{Applying \cref{lem:monotone_improvement_scalar} in either case yields}
\begin{equation}
    \label{eqn:half_plane}
    V(f(\dummy, \u^\dagger(\x_t), \d)) \leq V(f(\dummy, \u^*, \d))\ \forall \dummy \in \cvx,\ \forall \d \in \mDx.
\end{equation}
Finally, if $\mathcal{K}(\cvx) = \mathbb{R}$, then $\nabla V(f(\dummy, \u, \d))^\top g(\dummy) = 0\ \forall \dummy \in \cvx,\ \forall \u \in \mathcal{U},\ \forall \d \in \mDx$, so
\begin{equation}
    \label{eqn:all_reals}
    V(f(\dummy, \u^\dagger(\x_t), \d)) = V(f(\dummy, \u^*, \d))\ \forall \dummy \in \cvx,\ \forall \d \in \mDx.
\end{equation}
\revieww{Evaluating \cref{eqn:half_plane,eqn:all_reals} at the true state $\x_t$ and recalling that $\min_{\d \in \mDx} V(f(\x_t,\u^\dagger(\x_t),\d)) \geq 0$ results in $\min_{\d \in \mDx}V(f(\x_t, \u^*, \d)) \geq 0$, so $\x_{t+1}\in\Omega^*$.}
\end{proof}
Intuitively, this proof shows that at each time step, if the filter is active, \gdn\ pushes the true state in a safe direction, resulting in safety at the next time step.
Interestingly, this does not mean that $\min_{\dummy \in \bar{\mathcal{X}}_t}V(\dummy) \geq 0$ \review{or $\Phi(\bar{\mathcal{X}}_t,\u^*_t)\geq0$} at every time step.
\review{As shown in \cref{sec:numerical_results:taxinet}, these conditions may be violated due to a sudden increase in the size of $\bar{\mathcal{X}}_t$ as a result of} varying sensitivity to $\Delta\y$.
\review{However, safety is still guaranteed because $\x_t \in\bar{\mathcal{X}}_t \cap \Omega^*$, and \revieww{the proof establishes $\Phi(\bar{\mathcal{X}}_t \cap \Omega^*, \u^*_t) \geq 0$ even when $\Phi(\bar{\mathcal{X}}_t, \u^*_t) < 0$.}}


\subsection{Scalability}
There are two scalability considerations associated with this approach.
First, the scalability of the \review{online} NNV \review{calculation} with respect to NN size depends on the specific NNV tool used and can be expected to improve as verification methods advance.
In this work, we use \texttt{CROWN}~\cite{zhang_efficient_2018}, which has polynomial time complexity in both the number of layers and neurons per layer in $L_\theta$.
Second, the scalability limitations of HJ reachability are well known~\cite{bansal_hamilton-jacobi_2017}, since solving \cref{eqn:programming} typically requires discretizing the state space (\review{offline}), leading to exponential complexity in the state dimension.
\revieww{Given this bottleneck, \cref{eqn:gdn_ctrl} can be solved quickly at runtime using derivative-free methods.}
In future work, we plan to explore applying \gdn\ with learned~\cite{bansal_deepreach_2021} and latent-space~\cite{nakamura_generalizing_2025} reachability methods to improve scalability.
\section{Numerical Results}
\label{sec:numerical_results}
This section presents several examples demonstrating \gdn\ and its properties, including comparisons to several CBF-based approaches. All experiments were conducted on a machine running Ubuntu 22.04 with an i7-6700K CPU and 32 GB of RAM using \texttt{jax\_verify}~\cite{jaxverify} for NNV and the \texttt{hj\_reachability} python toolbox~\cite{jax_hj_reachability}.


\subsection{Taxinet Safety}
\label{sec:numerical_results:taxinet}
We first demonstrate \gdn\ on a runway taxiing problem similar to \cite{katz_verification_2022}.
Consider an aircraft with state $\x = [p_x, p_y, \theta]^\top$, where $p_y$ and $\theta$ represent the cross-track error (CTE) and the heading error (HE), measured from the centerline of the runway, respectively, and $p_x$ is the down-track position (DTP).
The dynamics of the aircraft are
\begin{equation*}
    \begin{split}
        & p_{x,t+1} = p_{x,t} + T_s\left(v\cos(\theta_t) + d_t^x\right) \\
        & p_{y,t+1} = p_{y,t} + T_s\left(v\sin(\theta_t) + d_t^y\right) \\
        & \theta_{t+1} = \theta_t + T_s\left(\frac{v}{\ell}\tan(\frac{\pi}{180}\phi_t) + d_t^\theta\right),
    \end{split}
\end{equation*}
where $v=5.0$~m/s is the constant velocity of the aircraft, $\ell=5.0$~m is the distance between its front and back wheels, $\phi_t \in [-12^\circ,12^\circ]$ is the rudder angle, $T_s=0.1$~s, $|d_t^x|,|d_t^y| \leq 0.02$~m, and $|d_t^\theta| \leq 0.01$~rad.
Realistic images of a runway scenario are obtained by simulating the dynamics in X-Plane~11, providing $8\times16$ resolution downsampled images which are passed into the NN estimator $L_\theta$ discussed in \cite{katz_verification_2022}.
There is a simple  nominal controller (proportional control) given by $K_d({{\x}}_t) = -0.74\review{p_{y,t}}-0.44\theta_t$.
The safe set is defined by the runway boundary, i.e., $\review{\mathcal{C}} \triangleq \{\x\ \vert\ |\review{p_{y,t}}| \leq 10\}$.

As presented in \cite{katz_verification_2022,lin_robust_2025}, under nominal conditions, the estimator $L_\theta$ and controller $K_d$ are guaranteed to keep the system safe under most initial conditions.
However, \cref{fig:taxinet_trajectories} shows a trajectory (dashed purple) demonstrating \review{how an aircraft without a filter may violate the safety constraints due to} an adversarial attack \review{$\Delta\y_t$ (calculated using projected gradient descent (PGD)~\cite{madry_towards_2019} with $\epsilon=0.021$)}.
In contrast, when the system uses \gdn\ from the same initial condition (solid black), it calculates $\bar{\mathcal{X}}_t$ (shaded teal) and uses the filter $\review{K_F(\bar{\mathcal{X}}_t, \u_t)}$ \cref{eqn:gdn} to successfully stay on the runway while achieving a total per-step calculation time of $0.032 \pm 0.003$~s.
Moreover, to validate \cref{thm:gdn_scalar}, we generated $100,000$ sample trajectory rollouts $\xi_i$ (light teal), where ${\xi_i = {\{ \x'_t \ \vert \ \x'_{t+1} = P_{\bar{\mathcal{X}}_t}(f(\x'_t, \u_t, \d_t)),\ \d_t \sim \review{\mDx}\ \forall t\in[0,T]\}_i}}$ with each $\{\x_{0}\}_i\sim \bar{\mathcal{X}}_0$ and where $P_{\bar{\mathcal{X}}_t}(\cdot)$ projects a state $\x'_t$ into ${\bar{\mathcal{X}}_t}$ to enforce that it is valid given the observation $\tilde{\y}_t$, i.e., $\exists\d^y\in \mathcal{D}^y$ such that $L_\theta(h(\x'_t, \d^y)) \in \bar{\mathcal{Y}}_t$. 
As shown in \cref{fig:taxinet_trajectories}, even though $\bar{\mathcal{X}}_t$ leaves the safe set between 233 and 236~m, each $\xi_i$ remains safe \review{because when $\Phi(\bar{\mathcal{X}}_t, \u_t) <0$, the filter commands $u_t = -12$ (a maximal right-hand turn), pushing all $\x'_t\in\bar{\mathcal{X}}_t$ in a safe direction.}
\begin{figure}[t]
    \centering
    \includegraphics[width=1.0\linewidth, trim={55 55 125 75pt},clip]{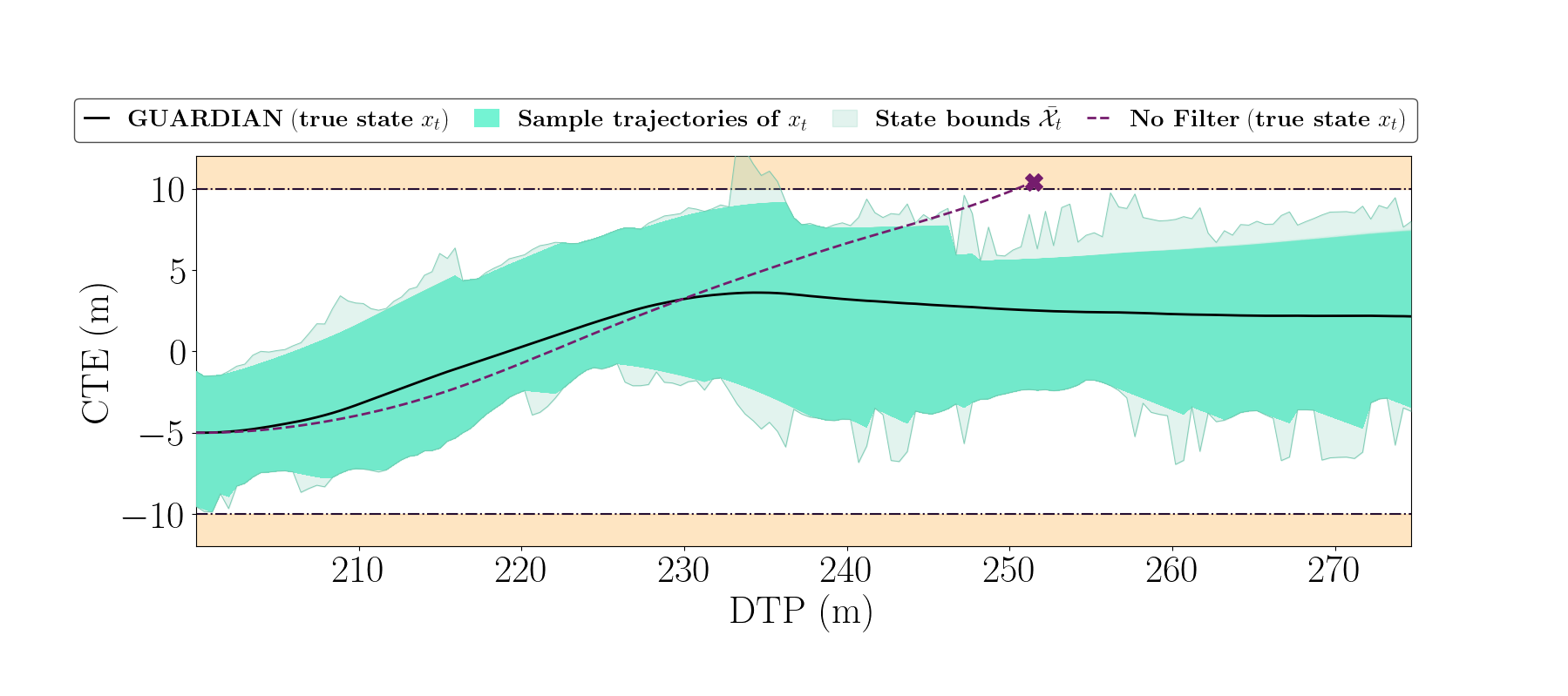}
    \caption{Trajectories of a taxiing aircraft subject to adversarially perturbed observations with and without \gdn.}
    \label{fig:taxinet_trajectories}
    \vspace{-20pt}
\end{figure}

\ifarxiv
\cref{fig:v_contours} shows the value function for the runway problem (DTP is irrelevant for safety, so $V(\x)$ is only a function of $x$ and $\theta$) along with a snapshot of $\bar{\mathcal{X}}_t$ from the .
\begin{figure}[t]
    \centering
    \includegraphics[width=0.7\linewidth, trim={10 0 50 50pt},clip]{Figures/value_function_snapshot_58_r1.png}
    \caption{Contours of $V(\x)$ for a taxiing aircraft.}
    \label{fig:v_contours}
\end{figure}
\fi



\subsection{State-Dependent Vulnerability to Adversarial Attacks}
\label{sec:state-dependent_attack}
This section shows how \gdn\ uses NNV to account for and protect against the state-estimate's sensitivity to measurement perturbations.
Consider the system with state vector $\x = [p_x, p_y, v_x, v_y]^\top$ and dynamics
\begin{equation}
\label{eqn:dynamics:double_integrator}
    \begin{split}
        & \p_{t+1} = \p_t + T_s\left(\v_t + \frac{1}{2}T_s\u_t + \d_t^p\right)\\
        & \v_{t+1} = \v_t + T_s(\u_t + \d_t^v),
    \end{split}
\end{equation}
where $\p = [p_{x}, p_{y}]^\top$, $\v = [v_{x}, v_{y}]^\top$, and $\u = [u_{x}, u_{y}]^\top$ with $\|\u_t \|_\infty \leq 5.0$,  $T_s = 0.1$, $\|\d_t^p\|_\infty \leq 1e^{-5}$, and $\|\d_t^v\|_\infty \leq 1e^{-5}$. \review{Although the system has two control inputs, it can be decomposed into two scalar-input subsystems, enabling separate synthesis of \gdn's safety filter for each.}
Range measurements are taken from a set of landmarks $\ell_i \in \mathbb{R}^2,\ i\in\{1,2,3,4\}$ such that for each time step, $\y_t \in \mathbb{R}^4$ has elements obtained via $\y_{t,i} = \|\ell_i -\p_t\| + \Delta \y_t$, where $\Delta \y_t$ is obtained via \review{PGD~\cite{madry_towards_2019} with $\epsilon = 0.05$}.
We use a learned estimator $L_\theta$, \review{with two hidden layers containing 256 neurons each,} trained to predict $\hat{\x}$ using a history of three measurements $\y_{t:t-2}$, i.e., $\hat{\x}_t = L_\theta(\y_{t:t-2})$.
The nominal controller $K_d(\tilde{\hat{\x}},t)$ is designed to track a circle-like trajectory (shown in \cref{fig:state_vulnerability}) within a $10\times10$~m region defining the safe set.
The attack $\Delta \y_t$ is designed to bring $\tilde{\hat{\x}}_t$ towards the center of the circle, resulting in the true state $\x_t$ trending radially outward toward the unsafe region.
\begin{figure}[t]
    \centering
    \includegraphics[width=0.95\linewidth]{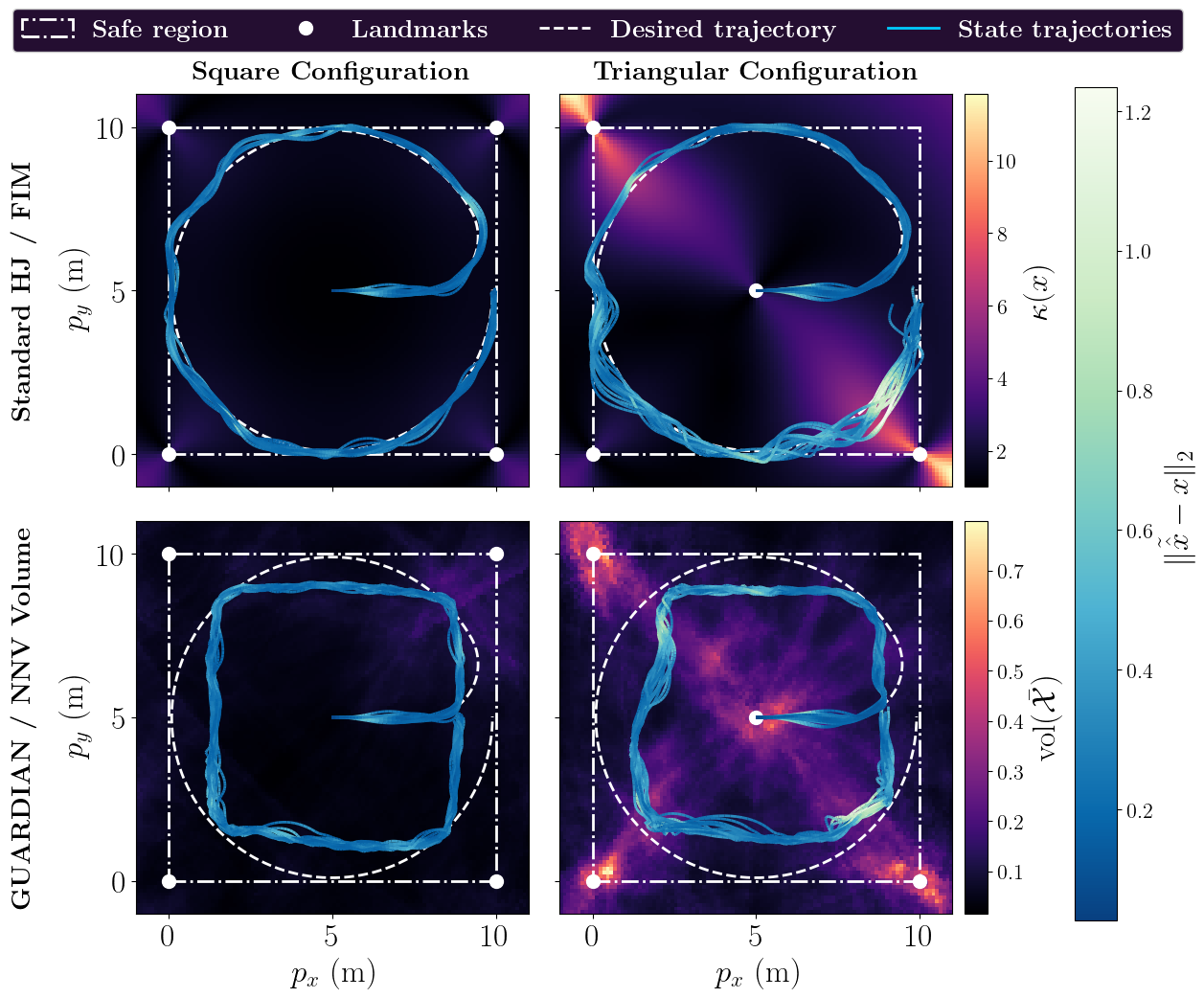}
    \caption{State \review{estimation} error for sample trajectories with different landmark configurations show state-dependent vulnerability to adversarial attacks.
    Triangular landmarks (right) are more vulnerable than square landmarks (left).
    Standard HJ reachability (top) does not ensure safety, but \gdn\ (bottom) does by quantifying state uncertainty.
    FIM condition number (top) is captured by NNV set volume (bottom).}
    \label{fig:state_vulnerability}
    \vspace{-20pt}
\end{figure}


Each panel of \cref{fig:state_vulnerability} shows twenty sample trajectories using standard HJ safety filtering~\cref{eqn:hj} (top row), and \gdn\ (bottom row) for square (left column) and triangular (right column) landmark configurations.
The trajectories are colored to show the amount of estimation error affecting the system \review{throughout each rollout}, indicating how effective the attack was in different regions.
When compared with the square landmark configuration on the left, the right plots show that the triangular landmark configuration results in trajectories with more error, especially near the top left and bottom right portions of their rollouts.
Interestingly, \gdn\ detects this increase in error and acts more conservatively in those regions, as shown by the deflection of the trajectories when compared with the square configuration.

The triangular configuration trajectories also show that the attack is more effective along the diagonal on which most of the landmarks are aligned.
\review{This is because there is redundant information along the diagonal and less information orthogonal to it.}
\review{To quantify this}, we consider an analog to the Fisher information matrix (FIM)~\cite{cintron-arias_sensitivity_2020}: 
$I(\x) = \frac{1}{\epsilon^2} J^\top J,$
where $J\triangleq \frac{\partial h}{\partial \x}$.
We can then quantify the quality of information received by a measurement at $\x$ via the condition number $\kappa(\x) \triangleq \frac{\lambda_1(I(\x))}{\lambda_2(I(\x))}$, where $\lambda_1(I(\x))$ and $\lambda_2(I(\x))$ are the first and second largest eigenvalues of $I(\x)$, indicating how sensitive a measurement at $\x$ is to perturbations.
The heatmap in \cref{fig:state_vulnerability} (top) shows the FIM condition number $\kappa(\x)$ for both configurations of landmarks.
Larger condition numbers indicate measurements with worse information (i.e., those more sensitive to perturbations), so the heatmaps confirm the previously described intuition by showing that the state estimate error increases in poorly conditioned regions.

For comparison to $\kappa(\x)$, the heatmaps in the  bottom plots of \cref{fig:state_vulnerability} show the volume of the state bounds, i.e., $\mathrm{vol}(\bar{\mathcal{X}})$, obtained using NNV with observations at each $\x$ subject to $\epsilon$ perturbations.
As shown by comparing the two right plots, NNV captures similar trends in measurement sensitivity to those calculated analytically with the FIM.
This capability enables \gdn\ to incorporate state-dependent sensitivity analysis into its filtering approach, resulting in a filter that is appropriately conservative depending on varying levels of vulnerability to adversarial attacks.
As a result, \gdn\ (bottom) acts as an effective safety filter, ensuring \review{safety for} all trajectories, whereas the standard HJ \review{filter} (top) does not.
\review{Notably, despite the higher dimensionality of this problem, \gdn\ achieves a total per-step calculation time of $0.0042 \pm 0.0003$~s, faster than the 2D case in \cref{sec:numerical_results:taxinet} due to the use of a smaller NN estimator.}

\subsection{Comparison to MR-CBFs, R-CBFs, and R-CBF-QPs}
\label{sec:comparisons}
Measurement-Robust CBFs (MR-CBFs)~\cite{dean_guaranteeing_2021}, Robust CBFs (R-CBFs)~\cite{nanayakkara_safety_2025},
and adaptive robust CBFs (R-CBF-QPs)~\cite{das_safe_2025} were not specifically designed for defense against adversarial attacks, but their handling of state uncertainty in safety-critical settings means they could be adapted for such purposes.
Thus, to highlight key differences between these approaches and \gdn, we consider the scalar system with dynamics $x_{t+1} = x_t + T_s(u_t + d_t^x)$ with $K_d(x,t) = x_{ref,t} - x_t$, and with a nonlinear observation function $y_t = \frac{1}{4}\mathrm{log}\left(\frac{4}{x_t+3}-1\right)+d_t^y$ where $|d_t^y| \leq \epsilon$, \review{$T_s=0.01$, $|d_t^x| \leq 0.01$}, and \review{$|u_t| \leq 1.0$}.
Solving for $x_t$ in terms of $y_t$ then gives ${L(y_t) \triangleq 4\cdot\mathrm{sigmoid}(4y_t)-3}$, which acts as a perfect estimator in the absence of noise.

\cref{fig:cbf_comparison:0.1} shows that \gdn\ is less conservative than both MR-CBFs and R-CBFs and satisfies the safety constraint where the R-CBF-QP does not.
For each filter, the initial state estimate is $\hat{x}_0 = 0$ and the disturbance term $d_t^y$ is set at a constant value of $\epsilon=0.1$, resulting in the true state $x_t$ always being at the lower bound of possible values given the measurement model and possible noise.
The reference value is set to $x_{ref}=-0.95$ for $t\in[0, 6)$ and $x_{ref}=0.95$ for $t\in[6, 15]$.
Like in \cref{sec:state-dependent_attack}, the shaded NNV state bounds $\bar{\mathcal{X}}_t$ about the state estimate $\hat{x}_t$ for \gdn\ (dashed teal) show that the measurements are much more sensitive to perturbations near the lower boundary.
Using the resulting \review{$\bar{\mathcal{X}}_t$}, \gdn\ allows $\hat{x}_t$ to converge to $x_{ref}$ when it is set near the upper boundary, but intervenes at the lower boundary to ensure $x_t$ (solid teal) stays in the safe set.

\begin{figure}[t]
    \centering
    \captionsetup[subfigure]{aboveskip=0pt,belowskip=0pt}
    \begin{subfigure}{\linewidth}
        \centering
        \begin{tikzpicture}
            \node[inner sep=0pt] (img) {
                \includegraphics[width=0.84\linewidth, trim={0 0 0 0cm},clip]{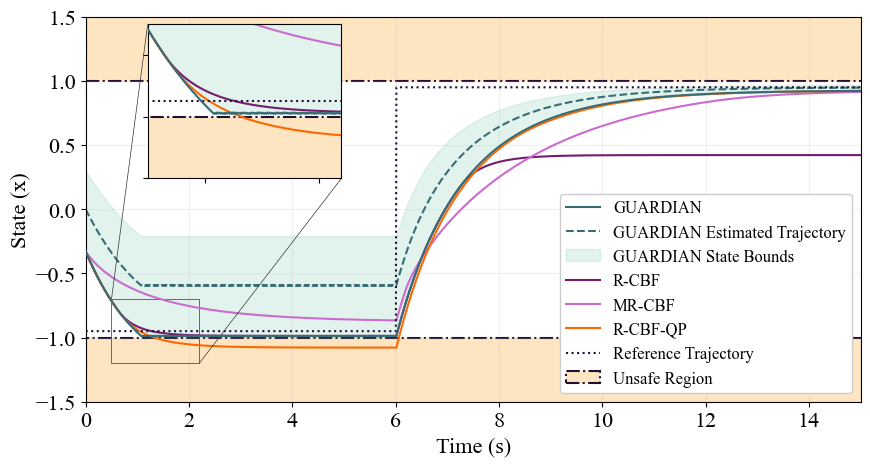}
            };
            \node[anchor=north east] at ([xshift=-4pt,yshift=-5pt]img.north east) {\(\epsilon = 0.1\)};
        \end{tikzpicture}
        \caption{\gdn\ is less conservative than MR-CBFs and R-CBFs, and safer than R-CBF-QPs.}
        \label{fig:cbf_comparison:0.1}
    \end{subfigure}
    \begin{subfigure}{\linewidth}
        \centering
        \begin{tikzpicture}
            \node[inner sep=0pt] (img) {
                \includegraphics[width=0.84\linewidth, trim={0 0 0 0cm},clip]{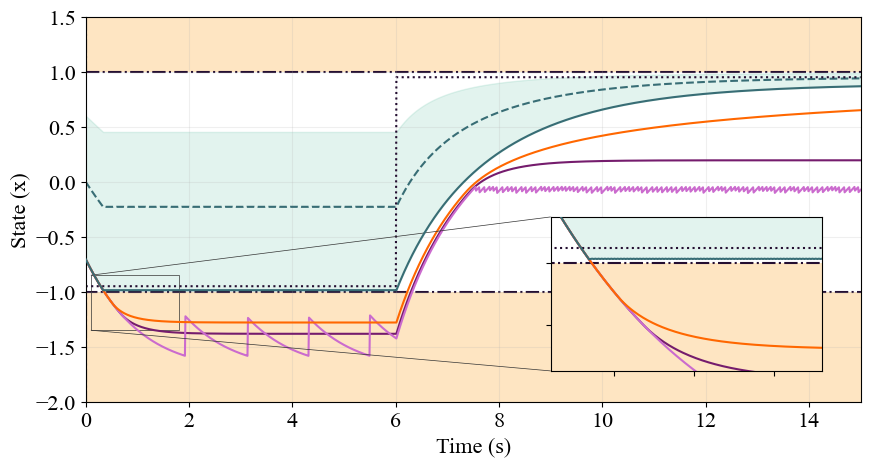}
            };
            \node[anchor=north east] at ([xshift=-4pt,yshift=-3.5pt]img.north east) {\(\epsilon = 0.2\)};
        \end{tikzpicture}
        \caption{\gdn\ effectively handles larger state estimate errors (i.e., from adversarial attacks) than MR-CBFs, R-CBFs, and R-CBF-QPs.}
        \label{fig:cbf_comparison:0.2}
    \end{subfigure}
    \label{fig:cbf_comparison}
    \vspace{-12pt}
    \caption{\gdn\ outperforms MR-CBFs, R-CBFs, and R-CBF-QPs in defense against adversarial attacks.}
    \vspace{-20pt}
\end{figure}

By comparison, R-CBFs assume a maximum state error and \review{thus are equally conservative at both boundaries despite the asymmetric sensitivity.
The adaptively tuned R-CBF-QP (orange) is less conservative, but violates safety at the lower boundary since it is invariant for an inflated version of $\mathcal{C}$ given large estimation errors.}
The MR-CBF (pink) uses the same state bounds as \gdn, but is more conservative in its transient and steady state behavior.

Moreover, as shown in \cref{fig:cbf_comparison:0.2}, each of the CBF-based approaches suffer from performance degradation in the presence of large estimation error.
Above a given threshold $\varepsilon_R$ (resp. $\varepsilon_Q$), R-CBFs (resp. R-CBF-QPs) provide set invariance for an inflated version of the safe set, but not necessarily the original safe set\review{, as shown when $\epsilon$ is increased to $0.2$ in \cref{fig:cbf_comparison:0.2}}.
Similarly, if the estimation error \review{exceeds} a threshold $\varepsilon_M$, no MR-CBF will exist, as shown by the MR-CBF trajectory \review{with failed optimizations in \cref{fig:cbf_comparison:0.2}}.
In contrast, \gdn\ effectively handles the increase in $\epsilon$, \review{appropriately adjusting} conservativeness \review{compared to \cref{fig:cbf_comparison:0.1}} to account for the increase in state uncertainty.
\review{In fact, for this problem, it can be shown analytically that \cref{asm:safe_cone_scalar} is satisfied on the neighborhood $\mathcal{N}(\partial\Omega^*) = (-\infty, -0.02] \cup [0.02, \infty)$ and $\iss = [-0.9899, 0.9899]$; therefore, \gdn\ guarantees safety as long as $\bar{\mathcal{X}}_t \subset (-\infty, 0.9899]$ or $\bar{\mathcal{X}}_t \subset [-0.9899, \infty)$.}
Details on the specification of $\varepsilon_M$, $\varepsilon_R$, and $\varepsilon_Q$ can be found in \cite{dean_guaranteeing_2021}, \cite{nanayakkara_safety_2025}, and \cite{das_safe_2025}, respectively, but \review{their} existence limits the applicability of these approaches in the presence of adversarial attacks, \review{where estimation error is typically larger than that present due to random noise or imperfect learning}.

\section{Conclusion}
\label{sec:conclusion}
This paper presented \gdn: a reachability-based safety filter that incorporates state uncertainty resulting from noisy measurements to defend ONFLs from adversarial attacks.
Specifically, \gdn\ uses advanced NNV tools to determine a set containing all possible true states given a known attack strength and finds a maximally safe control over the entire safe set.
We demonstrated various properties of \gdn\ with three sets of results, including a realistic runway taxiing problem with image observations, and a problem comparing our approach with existing measurement-robust approaches.
In future work, we will seek to improve the scalability of our approach by considering multidimensional inputs and incorporating modern reachability approaches that leverage learned and latent representations.

\bibliographystyle{aiaa}
\balance
\bibliography{refs}
\end{document}